\documentclass[11pt]{article}
\usepackage[utf8]{inputenc}
\usepackage[T1]{fontenc}
\usepackage{lmodern}
\usepackage{microtype}
\usepackage[a4paper,margin=1in]{geometry}
\usepackage{amsmath,amssymb,amsthm}
\usepackage{booktabs}
\usepackage{graphicx}
\usepackage{xcolor}
\usepackage{hyperref}
\hypersetup{colorlinks=true,linkcolor=blue,citecolor=blue,urlcolor=blue}

\newtheorem{theorem}{Theorem}

\newtheorem{proposition}{Proposition}
\newtheorem{corollary}{Corollary}
\theoremstyle{definition}
\newtheorem{definition}{Definition}
\theoremstyle{remark}
\newtheorem{remark}{Remark}

\title{RetractorDB: A Deterministic Edge Signal Processing Engine\\
Based on Rational Beatty Sequences and Fraenkel's Partition}

\author{Michal Widera\\
\small Independent Researcher, Zabrze, Poland\\
\small ORCID: \href{https://orcid.org/0000-0002-3578-3792}{0000-0002-3578-3792}\\
\small \href{mailto:michal@widera.com.pl}{michal@widera.com.pl}}
\date{July 2026}

\begin{document}

\maketitle

\begin{abstract}
We present RetractorDB, an open-source edge signal processing engine
(ESPE) for regular time series whose query semantics is grounded in the
number theory of covering systems. RetractorDB is designed to support, not
replace, time-series databases (TSDB) and data stream management systems
(DSMS): deployed close to the signal source, it pre-processes and filters
high-frequency measurements on the edge device through a declarative
signal-processing query language, maintains a partial, correctable record
of past and scheduled future events in inspectable artifacts, and transmits
exact, deterministic results upstream, so that only reduced,
already-processed streams reach the central architecture. The data
model is differential (a stream is a pair
$(s_n,\Delta)$ with a constant rational inter-arrival interval), and the
core rate-conversion operators, interleave and de-interleave, are proved to
be rational Beatty sequences satisfying the conditions of Fraenkel's
partition theorem. This yields an algebra in which resampling is an exact,
deterministic, first-class operator: de-interleaving inverts interleaving
bit-for-bit using rational arithmetic alone, and algebraic rewrite rules
license query-plan optimization without changing results. We describe the
end-to-end realization of this algebra in a working engine: declarative
query language (RQL), compilation to a dependency DAG with rational interval
resolution, slot-based runtime scheduling, and an inspectable artifact
format with schema and null/gap metadata. We validate the semantics on
deterministic query examples drawn from the engine's integration tests,
including a complete Pan--Tompkins QRS-detection pipeline over MIT-BIH ECG
data expressed entirely within the algebra. A performance evaluation under a
real-time operating environment is in progress and deferred to a subsequent
version.
\end{abstract}


\section{Introduction}
\label{sec:intro}

Continuous stream systems often face three hard requirements at once: low
latency, determinism, and mathematically grounded operator semantics.
RetractorDB addresses this combination by connecting number-theoretic
covering systems with a practical database query-processing stack.

The system perspective is intentionally simple: users declare source streams
and continuous transformations in RQL, compile and execute with
\texttt{xretractor}, inspect live outputs via \texttt{xqry}, and validate
binary artifacts with \texttt{xtrdb}. Unlike typical best-effort stream
runtimes, RetractorDB emphasizes deterministic execution order and
reproducible outputs: identical input traces produce identical artifacts.

\paragraph{Positioning: an edge companion to TSDB and DSMS.}
RetractorDB is best understood as an \emph{edge signal processing engine}
(ESPE), not as a replacement for a time-series database (TSDB) or a
general-purpose data stream management system (DSMS). In a typical
deployment it sits between the signal source and an
upstream store: it ingests regular streams at the edge, applies exact rate
conversion and signal operators, and manages a bounded, artifact-based
record of the stream. The primary operational motivation is offloading.
High-frequency sources such as biomedical, industrial, or acoustic sensors
saturate links and ingestion pipelines when shipped raw; because RQL
expresses complete signal-processing pipelines declaratively (filtering,
decimation, feature extraction; see the Pan--Tompkins pipeline of
Section~\ref{sec:ecg}), the processing happens where the data is recorded,
and only rate-reduced, already-processed streams need to be transmitted
and retained centrally. This substantially relieves the transmission,
ingestion, and storage layers of the overall architecture.
Because every stream carries a constant rational
inter-arrival interval $\Delta$, the slot timeline is known ahead of
execution; the artifact model can therefore materialize slots for
\emph{future} events before their payloads arrive, and its shadow mechanism
permits partial, non-destructive correction of \emph{past} events
(Section~\ref{sec:event-artifact} develops this into a concrete
fault-recording scenario). The same
artifacts serve as the transmission unit toward the upstream TSDB or DSMS,
which remains the system of record for long-term retention, compression,
and ad-hoc analytics. In this division of labor, RetractorDB manages and
stores data in transit (exactly, deterministically, and auditably),
while the upstream system manages it at rest.

The distinguishing design decision is the data model. Mainstream DSMS
model a stream as a multiset of timestamped tuples
$\langle s,\tau\rangle$ and build query semantics on
windows~\cite{arasu2006,kramer2009}. RetractorDB instead adopts a
differential model $(s_n,\Delta)$ with a constant rational rate per stream,
and derives its rate-aligning operators, interleave and de-interleave,
from the theory of Beatty sequences and Fraenkel's partition
theorem~\cite{beatty1926,fraenkel1969}. Because Fraenkel's generalization
holds for rational parameters, the entire construction is realizable exactly
on a computer: no floating-point approximation enters the semantics. The
need for such an algebra was first articulated in the context of
computer-aided fetal monitoring~\cite{widera2003}; the
formal bridge between covering systems and stream alignment was established
in peer-reviewed form in 2006~\cite{widera2006,widera2006pl}; this paper
brings that result, its proofs, and the engine built on it into well-indexed
circulation, and positions it against five neighboring research lines
(Section~\ref{sec:related}). The scope of this paper is semantics and system
architecture; a quantitative performance study under a real-time Linux
environment is underway and will be reported in follow-up work.

\paragraph{Contributions.}
\begin{itemize}
\item A formal semantics for exact stream resampling: proofs that the
interleave operation is a sequential covering and that de-interleaving is a
rational Beatty sequence satisfying Fraenkel's conditions, hence exactly
invertible (Section~\ref{sec:foundations}).
\item Operator properties (commutativity of the stream sum, an interleave
shift-matching identity, event-order perturbation) that serve as rewrite
rules for query-plan optimization (Section~\ref{sec:foundations}).
\item A gap analysis positioning this synthesis against number theory,
Beatty-based scheduling, multirate DSP, DSMS, and time-series management
systems (Section~\ref{sec:related}).
\item An engine report from the open-source implementation: three-binary
architecture, compilation pipeline, rational slot scheduling, and an
artifact format designed for auditability and upstream transmission
(Sections~\ref{sec:system} and~\ref{sec:pipeline}).
\item A reproducibility-first semantic validation protocol based on
executable query examples with observed deterministic outputs, including a
Pan--Tompkins QRS-detection pipeline over MIT-BIH data
(Section~\ref{sec:examples}).
\end{itemize}

\section{Formal Foundations}
\label{sec:foundations}

This section presents the formal skeleton of the stream algebra: the data
model, the connection between stream operators and the theory of covering
systems, and the theorems on which the correctness and optimization of query
plans rely. The construction deliberately remains within a single domain:
the rational numbers. This is not a stylistic choice but the central point:
Beatty's theorem requires irrational parameters, which are not representable
in a computer, whereas Fraenkel's generalization admits rational parameters.
The theorems below show that the interleave and de-interleave operations are
a special case of Beatty sequences satisfying Fraenkel's conditions, and
are therefore realizable \emph{exactly} using rational arithmetic alone.

\subsection{Data Model and Operator Signature}

\begin{definition}[Regular time series]
\label{def:stream}
A \emph{data stream} (equivalently, a \emph{regular time series}) is an
ordered pair
\begin{equation}
S := (s_n, \Delta),
\end{equation}
where $(s_n)_{n=0}^{\infty}$ is an ordered series of tuples and
$\Delta \in \mathbb{Q}_{>0}$ is the constant, rational time interval between
consecutive elements. Stream indices run from $0$ throughout this paper;
the element $s_n$ carries the implied timestamp $(n+1)\Delta$, i.e., the
first element arrives one full interval after the stream's origin.
\end{definition}

This differential model $(s_n,\Delta)$ is the key point of departure from
mainstream DSMS semantics, where a stream is a multiset of timestamped pairs
$\langle s,\tau\rangle$~\cite{arasu2006}. Here the timestamp is not carried
per tuple; it is implied by the position $n$ and the rate $\Delta$.

The algebra underlying the query language is
\begin{equation}
A_{\mathrm{rql}} ::= \bigl( (s_n,\Delta_s),\;
(\#, \&, \%, +, -, >, @) \bigr),
\end{equation}
with the correspondence between formal operators and query-language symbols
given in Table~\ref{tab:operators}.

\begin{table}[t]
\centering
\caption{Formal operators and their query-language symbols.}
\label{tab:operators}
\begin{tabular}{lcc}
\toprule
Operation & Formal symbol & RQL symbol \\
\midrule
Projection & $\pi$ & field list after \texttt{SELECT} \\
Selection & $\sigma$ & logical condition \\
Sum & $\Sigma$ & \texttt{+} \\
Difference & $\delta$ & \texttt{-} \\
Interleave (entangle) & $\varphi$ & \texttt{\#} \\
De-interleave and its complement & $\Theta$, $\sim\!\Theta$ & \texttt{\&}, \texttt{\%} \\
Aggregation/serialization (AGSE) & $\Psi$ & \texttt{@} \\
Time shift & $\tau$ & \texttt{>} \\
\bottomrule
\end{tabular}
\end{table}

\begin{definition}[Interleave]
\label{def:interleave}
For streams $A=(a_n,\Delta_a)$ and $B=(b_n,\Delta_b)$, the interleave
$\varphi(A,B)$ produces the stream $C=(c_n,\Delta_c)$ with, for $n \ge 0$,
\begin{equation}
\label{eq:interleave}
c_{n}=
\begin{cases}
b_{\,n-\lfloor n z \rfloor}, & \lfloor n z \rfloor = \lfloor (n+1) z \rfloor,\\[2pt]
a_{\lfloor n z \rfloor}, & \lfloor n z \rfloor \neq \lfloor (n+1) z \rfloor,
\end{cases}
\qquad
z = \frac{\Delta_b}{\Delta_a+\Delta_b},
\qquad
\Delta_c = \frac{\Delta_a \Delta_b}{\Delta_a+\Delta_b}.
\end{equation}
\end{definition}

\begin{definition}[De-interleave]
\label{def:deinterleave}
The de-interleave is defined by two complementary operators, both indexed
over $n \ge 0$: $\Theta$ recovering the left constituent stream,
\begin{equation}
\label{eq:deinterleave-a}
a_{n} = c_{\,n+ \left\lceil \frac{(n+1)\Delta_a}{\Delta_b} \right\rceil},
\qquad
\Delta_a = \frac{\Delta_c \Delta_b}{\lvert \Delta_c-\Delta_b\rvert},
\end{equation}
and $\sim\!\Theta$ producing the residue (right constituent) stream,
\begin{equation}
\label{eq:deinterleave-b}
b_{n} = c_{\,n+\left\lfloor \frac{n\Delta_b}{\Delta_a}\right\rfloor},
\qquad
\Delta_b = \frac{\Delta_c \Delta_a}{\lvert \Delta_c-\Delta_a\rvert}.
\end{equation}
\end{definition}

\begin{definition}[Time shift]
\label{def:shift}
For $S = (s_n, \Delta)$ and $m \in \mathbb{N}$, the time shift is
$\tau_m(S) := \bigl( (s_{n+m})_{n=0}^{\infty},\, \Delta \bigr)$: the
stream advanced by $m$ samples, i.e., by the duration $m\Delta$.
\end{definition}

\begin{definition}[Stream sum]
\label{def:sum}
For $A=(a_n,\Delta_a)$ and $B=(b_n,\Delta_b)$, the stream sum
$\Sigma(A,B)$ produces the stream of concatenated tuples
$C=(c_n,\Delta_c)$ with, for $n \ge 0$,
\begin{equation}
\label{eq:sum}
c_n =
\begin{cases}
\bigl( a_n,\; b_{\lfloor n\Delta_a/\Delta_b \rfloor} \bigr),
& \Delta_a \le \Delta_b,\\[2pt]
\bigl( a_{\lfloor n\Delta_b/\Delta_a \rfloor},\; b_n \bigr),
& \Delta_a > \Delta_b,
\end{cases}
\qquad
\Delta_c = \min(\Delta_a, \Delta_b):
\end{equation}
the faster stream drives the output rate, and each of its elements is
paired with the element occupying the co-indexed slot of the slower
stream.
\end{definition}

The remaining operators of $A_{\mathrm{rql}}$ play no role in the
rate-conversion theory developed below: projection $\pi$ and selection
$\sigma$ act tuple-wise as in relational algebra, the window operator
$\Psi$ is illustrated operationally in Section~\ref{sec:examples}, and the
stream difference $\delta$ undoes $\Sigma$ by removing the attributes
contributed by its second argument; their formal definitions appear
in~\cite{widera2006}.

The pair $(\varphi;\Theta,\sim\!\Theta)$ is intended to behave like
multiplication and division on $\mathbb{N}$: interleaving composes two
streams into one, de-interleaving recovers a constituent together with a
well-defined remainder. Theorems~\ref{thm:covering}
and~\ref{thm:fraenkel} below make this claim precise.

\subsection{Beatty Sequences and Fraenkel's Theorem}

\begin{definition}[Parameterized Beatty sequence]
For $\alpha,\alpha' \in \mathbb{R}$, define
\begin{equation}
\mathcal{B}(\alpha,\alpha') :=
\left(\left\lfloor \frac{n-\alpha'}{\alpha} \right\rfloor\right)_{n=1}^{\infty}.
\end{equation}
The parameters have a direct geometric reading: $\alpha$ is the density,
$1/\alpha$ the slope, $\alpha'$ the shift, and $-\alpha'/\alpha$ the
$y$-intercept.
\end{definition}

This single definition yields a family of sequences. The partition
results below always concern a \emph{pair} of its instances with
different parameters: the pair is written
$\mathcal{B}(\alpha,\alpha')$ and $\mathcal{B}(\beta,\beta')$, the latter
denoting the complementary member.

Beatty's theorem~\cite{beatty1926} guarantees that two such sequences
partition $\mathbb{N}$ when the parameters are irrational with
$1/p + 1/q = 1$. Fraenkel's theorem~\cite{fraenkel1969} generalizes the
partition criterion and, crucially for computability, admits rational
parameters. We restate it in the form used throughout this paper; an
accessible proof is given by O'Bryant~\cite{obryant2003}.

\begin{theorem}[Fraenkel, 1969]
\label{thm:fraenkel-orig}
The sequences $\mathcal{B}(\alpha,\alpha')$ and
$\mathcal{B}(\beta,\beta')$ partition $\mathbb{N}$ if and only if the
following five conditions hold:
\begin{enumerate}
\item $0 < \alpha < 1$;
\item $\alpha + \beta = 1$;
\item $0 \le \alpha + \alpha' \le 1$;
\item if $\alpha$ is irrational, then $\alpha' + \beta' = 0$ and
$k\alpha + \alpha' \notin \mathbb{Z}$ for all integers $k \ge 2$;
\item if $\alpha$ is rational, and $q \in \mathbb{N}$ is the least number
such that $q\alpha \in \mathbb{N}$, then
$\frac{1}{q} \le \alpha + \alpha'$ and
$\lceil q\alpha' \rceil + \lceil q\beta' \rceil = 1$.
\end{enumerate}
\end{theorem}

The proofs below rely on elementary properties of the floor and ceiling
functions, collected here for self-containment. For $x \in \mathbb{R}$ and
$C \in \mathbb{Z}$:
\begin{align}
\lfloor x \rfloor = \lceil x \rceil &\iff x \in \mathbb{Z},
\label{eq:fc1}\\
\lfloor x \rfloor + 1 = \lceil x \rceil &\iff x \in \mathbb{R}\setminus\mathbb{Z},
\label{eq:fc2}\\
\lfloor x + C \rfloor = \lfloor x \rfloor + C. &
\label{eq:fc3}
\end{align}
Additionally, for $a,b \in \mathbb{N}_{>0}$, the greatest common divisor
satisfies
\begin{equation}
\label{eq:gcd-cases}
\gcd(a,b) = b \iff \tfrac{a}{b} \in \mathbb{N},
\qquad\text{and otherwise}\qquad
1 \le \gcd(a,b) \le \min(a,b),
\end{equation}
which yields the disjoint case split used in the proof of
Theorem~\ref{thm:fraenkel}.

\subsection{Correctness of Interleave and De-interleave}

\begin{theorem}[Interleave is a sequential covering]
\label{thm:covering}
The interleave operation~\eqref{eq:interleave} provides a sequential
covering of both index sets of its argument streams: every element of $A$
and every element of $B$ is selected exactly once, in order, with no gaps
and no repetitions.
\end{theorem}

\begin{proof}
Since $0 < z < 1$, the increment
$d_n := \lfloor (n+1) z \rfloor - \lfloor n z \rfloor$ equals either $0$
or $1$ for every $n \ge 0$. The interleave~\eqref{eq:interleave} selects
an element of $B$ exactly at the steps with $d_n = 0$ (the equality
branch) and an element of $A$ exactly at the steps with $d_n = 1$.

Consider the $B$-selection index $x_n = n - \lfloor n z \rfloor$. Across
one step, $x_{n+1} - x_n = 1 - d_n$: it increases by exactly $1$ at each
$B$-selection step and is unchanged otherwise. Consequently, if $n < n'$
are two consecutive $B$-selection steps, then $x_{n'} = x_n + 1$; and the
first $B$-selection step is $n = 0$ (as $0 < z < 1$ gives
$\lfloor 0 \rfloor = \lfloor z \rfloor = 0$, i.e.\ $d_0 = 0$) with
$x_0 = 0$. The $B$-selections therefore use the indices $0, 1, 2, \dots$
in order, with no gaps and no repetitions.

Symmetrically, the $A$-selection index $\lfloor n z \rfloor$ increases by
exactly $1$ at each $A$-selection step ($d_n = 1$) and is unchanged
otherwise, and at the first such step its value is $0$ (all preceding
steps have $d = 0$). Hence the elements of $A$ are likewise selected
exactly once each, in order.
\end{proof}

\begin{theorem}[De-interleave satisfies Fraenkel's conditions]
\label{thm:fraenkel}
Let $a,b \in \mathbb{N}_{>0}$ represent the rational rate ratio of the
constituent streams, $\Delta_a/\Delta_b = a/b$. The two tuple-selection
sequences defining the de-interleave
operation~\eqref{eq:deinterleave-a}--\eqref{eq:deinterleave-b} are, up to
the index alignment made explicit in the proof, a special
case of Beatty sequences satisfying the conditions of
Theorem~\ref{thm:fraenkel-orig} for rational parameters. Consequently, they
partition $\mathbb{N}_0 := \mathbb{N} \cup \{0\}$, the index set of the
interleaved stream, and de-interleaving exactly inverts interleaving
using rational arithmetic alone.
\end{theorem}

\begin{proof}
\emph{Part 1 (reduction to Beatty form).}
The selection sequence of the de-interleave residue
in~\eqref{eq:deinterleave-b} is
\begin{equation*}
\left( n + \left\lfloor \frac{nb}{a} \right\rfloor \right)_{n=0}^{\infty}.
\end{equation*}
Its initial term ($n=0$) is $0$; the terms for $n \ge 1$ form the Beatty
part. Since $n \in \mathbb{N}$, property~\eqref{eq:fc3} gives
$n + \lfloor nb/a \rfloor = \lfloor n + nb/a \rfloor$, so we may seek
$\alpha,\alpha'$ with
\begin{equation*}
\left( \left\lfloor \frac{n-\alpha'}{\alpha} \right\rfloor \right)_{n=1}^{\infty}
=
\left( \left\lfloor n\,\frac{a+b}{a} \right\rfloor \right)_{n=1}^{\infty}.
\end{equation*}
Reading off slope and intercept: with shift $\alpha' = 0$ we obtain
$\alpha = a/(a+b)$, hence the selection sequence restricted to $n \ge 1$
is exactly
\begin{equation}
\label{eq:beatty-form}
\mathcal{B}\!\left( \frac{a}{a+b},\, 0 \right)
= \left( \left\lfloor n\,\frac{a+b}{a} \right\rfloor \right)_{n=1}^{\infty}.
\end{equation}

\emph{Part 2 (verification of the five conditions; residue sequence).}
We check the conditions of Theorem~\ref{thm:fraenkel-orig} for
$\alpha = a/(a+b)$, $\alpha' = 0$:
\begin{enumerate}
\item $0 < a/(a+b) < 1$ holds for all $a,b > 0$;
\item $\alpha + \beta = 1$ holds with $\beta = b/(a+b)$;
\item for $\alpha' = 0$ this reduces to condition 1;
\item vacuous, since $\alpha$ is rational;
\item the least $q$ with $q\alpha \in \mathbb{N}$ is
$q = (a+b)/\gcd(a,b)$; then $1/q \le \alpha + \alpha' = \alpha$ holds, and
$\lceil q\alpha' \rceil + \lceil q\beta' \rceil = 1$ with
$\alpha' = 0$ forces $\lceil q\beta' \rceil = 1$, i.e.\
$0 < \beta' \le \gcd(a,b)/(a+b)$. Every admissible value generates the
same sequence (the complement of~\eqref{eq:beatty-form} in $\mathbb{N}$
is unique), and we fix $\beta' = \gcd(a,b)/(a+b)$.
\end{enumerate}
The sequence complementary to~\eqref{eq:beatty-form} in the sense of
Fraenkel's conditions is therefore
$\mathcal{B}\bigl( b/(a+b),\, \gcd(a,b)/(a+b) \bigr)$.
Reindexed by $n \mapsto n+1$ so as to run from $n=0$, in step with the
selection sequences of Definition~\ref{def:deinterleave}, it reads
\begin{equation}
\label{eq:residue}
\left( \left\lfloor
\frac{(n+1) - \frac{\gcd(a,b)}{a+b}}{\frac{b}{a+b}}
\right\rfloor \right)_{n=0}^{\infty}.
\end{equation}
Expanding~\eqref{eq:residue},
\begin{equation*}
\left\lfloor
\frac{(n+1) - \frac{\gcd(a,b)}{a+b}}{\frac{b}{a+b}}
\right\rfloor
=
\left\lfloor n\frac{a}{b} + n + \frac{a}{b} + 1 - \frac{\gcd(a,b)}{b} \right\rfloor .
\end{equation*}
Matching this, termwise for $n \ge 0$, against the selection sequence of
the recovered stream in~\eqref{eq:deinterleave-a},
$\bigl(n + \lceil (n+1)a/b \rceil\bigr)_{n=0}^{\infty}$, and
extracting the integer part $n+1$ via~\eqref{eq:fc3}, the claim reduces
(after substituting $n$ for $n+1$, so that $n$ ranges over
$\mathbb{N}_{>0}$) to the identity
\begin{equation}
\label{eq:key-identity}
\left\lfloor n\frac{a}{b} - \frac{\gcd(a,b)}{b} \right\rfloor + 1
= \left\lceil n\frac{a}{b} \right\rceil ,
\qquad n \in \mathbb{N}_{>0}.
\end{equation}

\emph{Part 3 (case analysis for~\eqref{eq:key-identity}).}
By~\eqref{eq:gcd-cases} two disjoint cases cover the domain.

\emph{Case 1: $\gcd(a,b) = b$, i.e.\ $a/b \in \mathbb{N}$.}
Then $na/b \in \mathbb{N}$, so by~\eqref{eq:fc1}
$\lceil na/b \rceil = \lfloor na/b \rfloor$, and by~\eqref{eq:fc3}
$\lfloor na/b - 1 \rfloor + 1 = \lfloor na/b \rfloor$. Both sides
of~\eqref{eq:key-identity} coincide.

\emph{Case 2: $b \nmid a$, so $1 \le \gcd(a,b) < b$ and
$0 < \gcd(a,b)/b < 1$.}
If $na/b \notin \mathbb{Z}$, then by~\eqref{eq:fc2}
$\lceil na/b \rceil = \lfloor na/b \rfloor + 1$; since the fractional part
of $na/b$ is a nonzero multiple of $\gcd(a,b)/b$, hence at least
$\gcd(a,b)/b$, subtracting $\gcd(a,b)/b$ from $na/b$ cannot cross an
integer below $\lfloor na/b \rfloor$, so
$\lfloor na/b - \gcd(a,b)/b \rfloor = \lfloor na/b \rfloor$
and~\eqref{eq:key-identity} follows.
If $na/b \in \mathbb{Z}$, then $\lceil na/b \rceil = na/b$, and since
$0 < \gcd(a,b)/b < 1$,
$\lfloor na/b - \gcd(a,b)/b \rfloor = na/b - 1$, again
giving~\eqref{eq:key-identity}.

Thus the two selection sequences describing the de-interleave operation
are, up to the unit reindexing above, Beatty sequences satisfying
Fraenkel's conditions for rational parameters: the pair
\eqref{eq:beatty-form} and
$\mathcal{B}\bigl( b/(a+b),\, \gcd(a,b)/(a+b) \bigr)$
partitions $\mathbb{N}$, and together with the initial residue term $0$
from Part~1 the selections partition $\mathbb{N}_0$, the full index set
of the interleaved stream. The recovered and residue streams are
therefore exact.
\end{proof}

\begin{corollary}[Exact invertibility on rationals]
\label{cor:exact}
For streams with rational rates, $\Theta$ and $\sim\!\Theta$ recover the
constituent streams of $\varphi(A,B)$ exactly (bit-for-bit): no tuple is
lost, duplicated, or reordered relative to its constituent stream. The pair
$(\varphi;\Theta,\sim\!\Theta)$ thus behaves as multiplication/division, and
$(\Sigma;\delta)$ as addition/subtraction, on the set of regular time
series.
\end{corollary}

\begin{remark}[Implementation constraint]
\label{rem:rational}
The practical consequence of Theorem~\ref{thm:fraenkel} is a hard
engineering rule: the implementation must not leave the rational domain,
even transiently. An implicit cast of an intermediate result to floating
point violates the premises of the theorem. Materialization to
floating-point form must be deferred until an explicit floor or ceiling
operation is applied. The RetractorDB engine enforces this rule throughout
the compilation and execution pipeline.
\end{remark}

\subsection{Operator Properties Used in Plan Optimization}

The algebra yields rewrite rules applied by the query-plan optimizer. We
state three representative properties.

\begin{proposition}[Event-order perturbation]
\label{prop:order}
The order of elements in an interleaved stream does not, in general, reflect
the order of occurrence of the underlying events.
\end{proposition}

\begin{proof}[Proof (by counterexample)]
Consider streams $\mathit{Alfa} = (\{1,2,3,4,5,6,\dots\},\, 2)$ and
$\mathit{Epsilon} = (\{a,b,c,d,e,f,\dots\},\, 3)$. Then
$\varphi(\mathit{Epsilon},\mathit{Alfa})$ produces
$\mathit{Tau} = (\{1,2,a,3,b,4,5,c,6,d,\dots\},\, 6/5)$. In $\mathit{Tau}$
the tuple $c$ appears \emph{after} the tuple $5$, yet $c$ occurs in
$\mathit{Epsilon}$ at second $9$ while $5$ occurs in $\mathit{Alfa}$ at
second $10$. Hence any analysis with respect to embedded event time must
first apply de-interleaving to recover the constituent streams.
\end{proof}

\begin{proposition}[Commutativity of the sum]
\label{prop:sum-comm}
Up to attribute order, the stream sum is commutative:
$\Sigma(A,B) \equiv \Sigma(B,A)$.
\end{proposition}

\begin{proof}
Assume $\Delta_a \le \Delta_b$; the opposite case is symmetric. The first
case of Definition~\ref{def:sum} gives
$\Sigma(A,B)_n = ( a_n,\, b_{\lfloor n\Delta_a/\Delta_b \rfloor} )$, while
for $\Sigma(B,A)$ the roles of the arguments are exchanged and its second
(or, when $\Delta_a = \Delta_b$, first) case applies, giving
$\Sigma(B,A)_n = ( b_{\lfloor n\Delta_a/\Delta_b \rfloor},\, a_n )$; both
carry $\Delta_c = \Delta_a$. The two streams coincide up to the order of
the concatenated attributes.
\end{proof}

\begin{theorem}[Interleave shift-matching]
\label{thm:shift-match}
The interleave is not commutative in general: since $0 < z < 1$, the
equality branch of~\eqref{eq:interleave} always holds at $n = 0$, so
$\varphi(A,B)$ begins with $b_0$ whereas $\varphi(B,A)$ begins with
$a_0$. It is, however, equivariant under rate-matched time shifts: if
$i,k \in \mathbb{N}$ are chosen with $i\Delta_a = k\Delta_b$ (both
arguments shifted by the same duration), then
\begin{equation}
\label{eq:shift-match}
\varphi\bigl( \tau_i(A),\, \tau_k(B) \bigr)
= \tau_{i+k}\bigl( \varphi(A, B) \bigr).
\end{equation}
\end{theorem}

\begin{proof}
Both sides carry the interval $\Delta_c$ of~\eqref{eq:interleave}, so it
suffices to compare elements. From $i\Delta_a = k\Delta_b$,
\begin{equation*}
(i+k)\,z = (i+k)\,\frac{\Delta_b}{\Delta_a+\Delta_b}
= \frac{i\Delta_b + i\Delta_a}{\Delta_a+\Delta_b} = i \in \mathbb{N},
\end{equation*}
hence, by~\eqref{eq:fc3}, for every $n \ge 0$ and $m := n+i+k$,
\begin{equation*}
\lfloor m z \rfloor = \lfloor n z + i \rfloor = \lfloor n z \rfloor + i .
\end{equation*}
In particular $\lfloor m z \rfloor = \lfloor (m+1) z \rfloor$ if and only
if $\lfloor n z \rfloor = \lfloor (n+1) z \rfloor$: position $n$ of the
left-hand side of~\eqref{eq:shift-match} and position $m$ of
$\varphi(A,B)$ take the same branch of~\eqref{eq:interleave}. On the
equality branch, the right-hand side selects
\begin{equation*}
b_{\,m - \lfloor m z \rfloor}
= b_{\,n+i+k - \lfloor n z \rfloor - i}
= b_{(n - \lfloor n z \rfloor) + k},
\end{equation*}
which is precisely the element $\bigl(\tau_k(B)\bigr)_{\,n-\lfloor n z
\rfloor}$ selected by the left-hand side; on the inequality branch it
selects $a_{\lfloor m z \rfloor} = a_{\lfloor n z \rfloor + i} =
\bigl(\tau_i(A)\bigr)_{\lfloor n z \rfloor}$, again matching the
left-hand side. The two streams coincide elementwise.
\end{proof}

These properties are not formalities: complementarity of $\varphi/\Theta$
and $\Sigma/\delta$ (Theorems~\ref{thm:covering}--\ref{thm:fraenkel}),
exactness on rationals (Corollary~\ref{cor:exact}), and the rewrite rules
(Propositions~\ref{prop:order}--\ref{prop:sum-comm},
Theorem~\ref{thm:shift-match}) are the license
under which the optimizer transforms query plans into cheaper equivalent
forms without changing results, and the reason the engine's outputs are
replay-stable. Full derivations in their original form appear
in~\cite{widera2006,widera2006pl}; a numerical verification of all equations
over exact rationals (Python \texttt{Fraction} prototypes) is available in
the project repositories~\cite{rdb_docs}.

\section{Related Work}
\label{sec:related}

The problem addressed by RetractorDB does not belong wholly to any single
discipline; it sits at the intersection of five. For each neighboring line
of research we answer three questions: what it has already solved, how
RetractorDB differs, and what it does \emph{not} address. The superposition
of the five answers delineates the gap this work fills
(Table~\ref{tab:gap}).

\subsection{Number Theory: Beatty Sequences and Covering Systems}

The algebra of Section~\ref{sec:foundations} rests on Beatty
sequences~\cite{beatty1926} and Fraenkel's rational-domain
generalization~\cite{fraenkel1969}. Beatty sequences have a rich
combinatorial literature and documented applications in aperiodic tilings,
periodic scheduling, digital lines in computer vision, and formal language
theory~\cite{shallit1997}. The line remains active: Schaeffer, Shallit, and
Zorcic showed that inhomogeneous Beatty sequences with quadratic irrational
slope are automaton-synchronizable, yielding decidability of their
first-order theory~\cite{schaeffer2024}. Most relevant here is the work of
Berger, Felzenbaum, and Fraenkel on disjoint covering systems of
\emph{rational} Beatty sequences~\cite{berger1986}, precisely the variant
underlying the de-interleave operation of Theorem~\ref{thm:fraenkel}.

\emph{What this line does not address:} number theory studies these
sequences as mathematical objects. It does not connect them to a database
engine, a stream-processing model, or signal processing. It supplies the
bricks, not the building.

\subsection{Beatty-Based Scheduling (Pinwheel)}

This is the closest neighboring application area at the level of proof
machinery, and we discuss it with particular care. In pinwheel scheduling,
tasks with different repetition periods are assigned to time slots belonging
to complementary Beatty sequences~\cite{eppstein2023}. Recent work on
pinwheel scheduling with real periods conducts its proofs on
Rayleigh/Beatty partitions using floor/ceiling identities of the form
$\lceil (m+l)\alpha \rceil - \lceil m\alpha \rceil$~\cite{pinwheel2025},
essentially the same apparatus as the proof of Theorem~\ref{thm:fraenkel}.
We read this in two ways. It is independent confirmation that the approach
is sound and natural. It also narrows what may be claimed as novel:
``Beatty sequences for scheduling'' exists and is actively published.
Notably, RetractorDB uses this machinery \emph{internally} for slot
scheduling as well (Section~\ref{sec:pipeline}), but that is not where
the original contribution lies.

\emph{What this line does not address:} scheduling treats the sequences as a
tool for allocating time slots to processors. It does not build a data
algebra on them, does not express signal operations with them, and does not
define a query language.

\subsection{Multirate DSP: Nonuniform Sampling and Filter Banks}

In DSP terms, interleave and de-interleave are sample-rate conversions
between streams of different $\Delta$. A mature literature exists. The
closest bridge is the characterization of nonuniform perfect-reconstruction
filter banks via unit-step responses by Samadi, Ahmad, and
Swamy~\cite{samadi2004}, which introduces step-function (implicitly floor)
machinery into multirate DSP. The broader line covers periodic nonuniform
sampling of bandlimited signals~\cite{margolis2008} and, directly
pertinent, perfect-reconstruction filter banks with \emph{rational}
sampling factors by Kova\v{c}evi\'{c} and Vetterli~\cite{kovacevic1993}.
Number-theoretic constructions do appear in this literature: Ramanujan
sums have been applied to periodicity extraction and signal
recovery~\cite{kalra2025}. However,
to the best of our knowledge, neither Beatty sequences nor Fraenkel's
theorem has been used in this literature, and this is part of the gap.

\emph{What this line does not address:} DSP operates in the $z$-domain and
frequency domain, on frames and bases. It does not cast resampling as a
declarative algebraic operator, nor embed it in a database system. The
coefficients may be rational, but the apparatus is analysis, not the number
theory of set partitions.

\subsection{Data Stream Management Systems (DSMS)}

On the database side, the canon is CQL from the Stanford STREAM
project~\cite{arasu2006}, where a stream is a potentially unbounded multiset
of elements $\langle s,\tau\rangle$ and query semantics is built on windows
and stream--relation mappings. The second close neighbor is the temporal
algebra of Kr\"{a}mer and Seeger (the PIPES system), which provides
deterministic results for continuous queries and a rich set of
transformation rules for optimization~\cite{kramer2009}. This is the proper
reference point for the algebra of Section~\ref{sec:foundations} and its
rewrite rules. The difference, however, is fundamental and concerns the data
model itself: CQL and PIPES build semantics on the timestamped model
$\langle s,\tau\rangle$, with per-tuple timestamps and window-based
operators. RetractorDB adopts the differential model $(s_n,\Delta)$ with a
constant rational $\Delta$ per stream, and derives the rate-aligning
operators from number theory. This is not a cosmetic difference in syntax;
it is a different data model, leading to a different operator class
(interleave, de-interleave) and a different optimization method. In
deployment terms the relationship is complementary rather than competitive:
RetractorDB acts as an edge-side pre-processing and buffering stage whose
exact, deterministic outputs can feed a window-based DSMS.

\emph{What this line does not address:} DSMS target approximate, scalable
processing of unbounded streams with tolerance for temporal disorder. They
do not aim at exact, deterministic DSP operations under strict timing
discipline, and they do not draw on number theory for resampling semantics.

\subsection{Time Series Management Systems and In-Database DSP}

This is the narrowest niche and the closest to RetractorDB's purpose. The
canonical survey is Jensen, Pedersen, and Thomsen~\cite{jensen2017}. The
Plato system described there is the nearest true ``DSP inside the
database,'' coupling an RDBMS with signal-processing methods to avoid
exporting data to external tools. Other ``signals in the database''
approaches reduce to approximation and compression: wavelet, dictionary,
and shape-based representations. All of them, however, treat DSP as
approximation or after-the-fact analytics. None makes signal-processing
operations \emph{exact, deterministic first-class operators} inside the
query algebra. RetractorDB does not attempt to compete on the TSMS axes of
ingestion scale, compression, and retention; as an ESPE it runs at the
edge, in front of a TSMS or TSDB, delivering exact signal-processing
results and correctable partial records for that upstream store to retain
at scale.

\emph{What this line does not address:} TSMS optimize ingestion scale,
compression, and retention. DSP is a second-class citizen in them: an
analytical add-on, not the core of the semantics.

\subsection{The Gap}

\begin{table}[t]
\centering
\caption{Coverage of the four defining properties across neighboring
research lines. No existing line occupies the intersection.}
\label{tab:gap}
\begin{tabular}{lcccc}
\toprule
Research line & Beatty/Fraenkel & Exact DSP &
\begin{tabular}{@{}c@{}}Stream algebra /\\query language\end{tabular} &
\begin{tabular}{@{}c@{}}Deterministic\\timing discipline\end{tabular} \\
\midrule
Number theory & \checkmark & -- & -- & -- \\
Scheduling (pinwheel) & \checkmark & -- & -- & partial \\
Multirate DSP & -- & \checkmark & -- & -- \\
DSMS (CQL, PIPES) & -- & -- & \checkmark & -- \\
TSMS / in-DB DSP & -- & partial & partial & -- \\
\midrule
\textbf{RetractorDB} & \checkmark & \checkmark & \checkmark & \checkmark \\
\bottomrule
\end{tabular}
\end{table}

Superimposing the five layers makes the picture legible: each discipline
touches one or two walls of the problem, but none occupies their
intersection. The contribution of RetractorDB lies not in any single
component but in their synthesis: the use of covering systems (rational
Beatty sequences and Fraenkel's theorem) as the semantic foundation of a
declarative stream algebra that realizes exact signal-processing operators
inside a database engine under a deterministic timing discipline. Number
theory has Beatty sequences and even scheduling, but connects them to
neither databases nor DSP. DSP has multirate processing and rational filter
banks, but does not reach for Fraenkel and does not cast this as a query
language. DSMS have stream algebras and optimization rules, but on the
windowed model $\langle s,\tau\rangle$, not the differential
$(s_n,\Delta)$. The intersection is empty.

\paragraph{Priority.}
The underlying problem, together with the need for a declarative stream
algebra and a continuous query language, was articulated in 2003--2005 in
the setting of computer-aided fetal monitoring~\cite{widera2003}; the
monitoring systems of that period, however, were implemented entirely by
hand, as neither the formal foundation nor the technology for a
declarative realization was available. The bridge ``covering systems
$\leftrightarrow$ stream alignment and DSP''
was then established in peer-reviewed form in
2006~\cite{widera2006,widera2006pl},
but in a venue of low international discoverability; its citations to date
circulate mainly in the Polish literature. Meanwhile, the scheduling
community has been publishing the same Beatty/Fraenkel machinery in
2023--2025~\cite{eppstein2023,pinwheel2025}. Part of the present paper's
purpose is to place that bridge, together with its proofs
(Section~\ref{sec:foundations}), into well-indexed circulation.

\paragraph{Methodological caveat.}
This is a directed review, not a systematic one: it is based on targeted
search within the five lines above, not on a complete citation analysis.
A forward-citation sweep of~\cite{samadi2004} supports the claim: as of
July 2026 (per Semantic Scholar), its only recorded citations are a
Gabor-window design paper, two multirate system-theoretic papers, and the
2006 bridge~\cite{widera2006} itself; none of them uses Beatty sequences
or Fraenkel's theorem. The nearest use of this machinery outside number
theory known to us is the construction of exponential Riesz bases from
Beatty--Fraenkel sequences by Pfander, Revay, and
Walnut~\cite{pfander2024}; it belongs to pure harmonic analysis and does
not touch filter banks or sample-rate conversion. A full treatment would
still add a systematic sweep of the pinwheel line~\cite{eppstein2023} and
of the filter-bank literature at large; if a use of Fraenkel's theorem in
multirate DSP exists, it would narrow the novelty claim accordingly and
should be incorporated here.

\section{System Overview}
\label{sec:system}

\subsection{Three-Binary Architecture}
RetractorDB is implemented in C++ as three cooperating programs:
\begin{itemize}
\item \textbf{\texttt{xretractor}}: parser, compiler, and runtime for RQL
plans, including compile-only and plan-dump modes;
\item \textbf{\texttt{xqry}}: client querying running streams via
inter-process communication, with raw and formatted output modes;
\item \textbf{\texttt{xtrdb}}: binary artifact inspector and editor used for
debugging, deterministic test-data generation, and result verification.
\end{itemize}

\subsection{Data and Artifact Model}
Artifacts are represented as binary payload files with associated metadata:
a \texttt{.desc} schema descriptor, a \texttt{.meta} null/gap index
(run-length-oriented metadata), and an optional \texttt{.shadow} file
supporting non-destructive overwrites. This separation supports compact
storage while preserving queryable structural information, and provides a
practical audit trail for null propagation, transmission gaps, and
historical-correction behavior. Together with the deterministic execution
model, it makes outputs directly comparable across runs: identical input
traces yield byte-identical artifacts, a property exercised continuously in
the project's integration test suite.

Two consequences of this model matter for the edge role of the engine.
First, because the slot timeline of every stream is fixed by its rational
$\Delta$, artifact slots for future events exist deterministically before
their payloads arrive; an artifact may therefore hold a \emph{partial}
record spanning both past and scheduled future events, with late payloads
and historical corrections applied non-destructively through the shadow
mechanism. Second, artifacts are self-describing (payload plus
\texttt{.desc} schema and \texttt{.meta} null/gap index) and therefore
directly transmissible: an upstream TSDB or DSMS can ingest them together
with their structural metadata. The artifact is thus both the storage
format and the transmission format: the concrete mechanism by which the
engine manages and stores data on behalf of the upstream system of
record.

\subsection{Event Artifacts: A Deployment Scenario}
\label{sec:event-artifact}

To make the edge role concrete, consider a fault-monitoring deployment:
the classical \emph{disturbance recorder} pattern known from power-grid
protection and cardiac event recording. A high-frequency source (a
vibration or current sensor, declared with an exact rational interval) is
recorded continuously at the edge. A detection chain performs FIR
bandpass filtering, rectification, and envelope thresholding, using the
same window-operator (\texttt{@}) and stream-sum convolutions as the
Pan--Tompkins pipeline of Section~\ref{sec:ecg}. A \texttt{RULE}
statement raises a trigger when the envelope crosses its threshold at
some slot~$t$.

The trigger produces an \emph{event artifact}: a cut of the recorded
stream spanning $[\,t - T_{\mathrm{pre}},\; t + T_{\mathrm{post}}\,]$,
where $T_{\mathrm{pre}}$ may range from seconds to hours of pre-fault
history, bounded only by edge retention. The two halves of this interval
show the two consequences noted in the preceding subsection in action.
The \emph{past} half is a slice of data already recorded; it is
available immediately, including any non-destructive corrections applied
through the shadow mechanism. The \emph{future} half exists at trigger
time as a deterministic slot timeline, because the stream's rational
$\Delta$ fixes every future slot in advance. The event artifact can
therefore be opened at the moment of detection with its full time axis
already laid out; it fills as payloads arrive and closes after
$T_{\mathrm{post}}$. The completed artifact (payload, \texttt{.desc}
schema, and \texttt{.meta} null/gap index) is then transmitted as a
self-describing unit to the upstream TSDB or DSMS, where the actual
fault analysis takes place with full pre- and post-fault context, while
the raw high-frequency stream never leaves the edge device.

\paragraph{Variant: adaptive-resolution transmission.}
The scenario extends naturally to steady-state bandwidth management.
Rather than transmitting nothing until an event occurs, the engine can
feed the upstream store continuously at reduced resolution: the
de-interleave operator (\texttt{\&}) extracts a low-rate substream of the
recording and transmits it as ordinary telemetry at a fraction of the
source bandwidth. Crucially, this substream is not a lossy preview but an
exact component of the original signal. When the detection chain
recognizes an event, the engine ships the complementary substream
(\texttt{\%}) for the window
$[\,t - T_{\mathrm{pre}},\; t + T_{\mathrm{post}}\,]$. By
Theorem~\ref{thm:fraenkel}, the two substreams partition the original
stream with neither gaps nor duplicates, so the upstream system recovers
the full-resolution signal around the event \emph{exactly}: it
interleaves the samples it has been receiving all along with those that
arrive after the trigger, bit-for-bit, with no resampling error. The link
thus carries low-rate data in steady state and full-rate data only for
the intervals that matter, while the receiving TSDB retains a
mathematically exact record of both regimes.

This subsection is an illustrative deployment scenario, not a validated
example in the sense of Section~\ref{sec:examples}. The detection chain
reuses operator constructions already demonstrated deterministically in
Section~\ref{sec:ecg}, the artifact properties it relies on are those
of the preceding subsection, and the adaptive-resolution variant rests
solely on the interleave/de-interleave exactness proved in
Section~\ref{sec:foundations}. End-to-end integration tests of the
triggered artifact cut and of split-path transmission are planned engine
work.

\subsection{Query Language}
RQL realizes the algebra of Section~\ref{sec:foundations}
(Table~\ref{tab:operators}) as a declarative language with
\texttt{DECLARE} (source streams with rational rates), \texttt{SELECT}
(continuous transformations), and \texttt{RULE} (alerting) statements.
Following Remark~\ref{rem:rational}, the engine keeps all rate and index
computations in exact rational arithmetic; floating-point materialization
occurs only at explicit output boundaries.

\section{Compilation and Runtime Pipeline}
\label{sec:pipeline}

\subsection{Compilation Stages}
Compilation comprises parsing, symbol expansion, alias and unfold handling,
type unification, dependency-graph construction, and interval resolution.
The result is a dependency DAG (query tree) with a resolved rational
interval per stream. Compile-only mode and DOT output are available for plan
inspection.

\subsection{Runtime Scheduling}
Execution traverses active nodes according to a rational timeline and
broadcasts outputs. Runtime progression selects the next slot from interval
counters, conceptually
\begin{equation}
t_k = \min_{\delta \in \mathcal{S}} \bigl( \delta \cdot \mathrm{counter}[\delta] \bigr),
\end{equation}
where $\mathcal{S}$ is the reduced set of primitive intervals. Because slot
generation follows the covering-system structure of
Section~\ref{sec:foundations}, activations are neither dropped nor
duplicated under rational timing.

\subsection{Determinism}
The project's testing discipline states the determinism expectation
explicitly: identical inputs must produce identical artifact outputs. The
\texttt{xtrdb} ecosystem makes this property inspectable and automatable in
continuous integration; the examples of Section~\ref{sec:examples} are drawn
directly from that suite.

\section{Deterministic Query Examples}
\label{sec:examples}

This section validates the semantics with executable examples: each shows
query text together with observed deterministic outputs, sourced from the
engine's integration tests and documentation. Performance-oriented
evaluation is deferred to a subsequent version (see
Section~\ref{sec:limitations}).

\subsection{Example A: Arithmetic Merge (\texttt{+})}
Source: \path{test/IntegrationTest_serial/simple/query.rql} and
\path{pattern-run.txt}.

\begin{verbatim}
STORAGE 'temp'

DECLARE a INTEGER STREAM core0, 0.1 FILE 'datafile2.dat'
DECLARE b INTEGER STREAM core1, 0.2 FILE 'datafile3.dat'

SELECT str1[0]*10,str1[1]*10,str1[1]*str1[0]+20 STREAM str1 FROM core0+core1
\end{verbatim}

Observed output excerpt:
\begin{verbatim}
9 Record(s)
12 Byte(s) per record.
{ INTEGER str1_0 INTEGER str1_1 INTEGER str1_2 }
{ str1_0:100 str1_1:10 str1_2:30 }
{ str1_0:110 str1_1:10 str1_2:31 }
{ str1_0:120 str1_1:20 str1_2:44 }
...
{ str1_0:170 str1_1:40 str1_2:88 }
{ str1_0:100 str1_1:50 str1_2:70 }
\end{verbatim}

This case illustrates deterministic arithmetic over merged streams with
different rates (the sum operator $\Sigma$ of
Section~\ref{sec:foundations}).

\subsection{Example B: Consistent Multi-Read from One Source}
Source: files \path{query-consitency.rql} and \path{pattern_xtrdb.txt}
in \path{test/IntegrationTest_serial/consistency/}.

\begin{verbatim}
DECLARE a INTEGER STREAM core0, 0.1 FILE 'datafile1.txt'

SELECT core[0]+100 STREAM str1 FROM core0
SELECT core[0]+200 STREAM str2 FROM core0
SELECT str2[0]-str1[0] STREAM str3 FROM str1+str2
\end{verbatim}

Observed output excerpt:
\begin{verbatim}
str1 - stream dump
{ str1_0:160 }
{ str1_0:161 }
{ str1_0:162 }
{ str1_0:163 }

str2 - stream dump
{ str2_0:260 }
{ str2_0:261 }
{ str2_0:262 }
{ str2_0:263 }

str3 - stream dump - expected all the same
{ str3_0:100 }
{ str3_0:100 }
{ str3_0:100 }
{ str3_0:100 }
\end{verbatim}

The invariant $\mathit{str3}=100$ confirms consistent read semantics when
one source feeds multiple downstream queries.

\subsection{Example C: Window Operator AGSE (\texttt{@}) with Null Handling}
Source: \path{test/IntegrationTest_serial/agse1/query.rql} and
\path{pattern.txt}.

\begin{verbatim}
DECLARE b INTEGER STREAM core1, 0.1 FILE 'datafile.txt'

SELECT * STREAM signalText0 FROM core1@(1,4)
SELECT * STREAM signalTextR FROM core1@(1,-4)
SELECT * STREAM signalText1 FROM core1@(2,4)
SELECT * STREAM signalText3 FROM core1@(2,2)
SELECT * STREAM signalText5 FROM signalText3@(1,1)
\end{verbatim}

Observed output excerpt (forward and reverse windows):
\begin{verbatim}
SELECT * STREAM signalText0 FROM core1@(1,4)
{ core1_0:10 core1_1:null core1_2:null core1_3:null }
{ core1_0:11 core1_1:10 core1_2:null core1_3:null }
{ core1_0:12 core1_1:11 core1_2:10 core1_3:null }

SELECT * STREAM signalTextR FROM core1@(1,-4)
{ core1_0:null core1_1:null core1_2:null core1_3:10 }
{ core1_0:null core1_1:null core1_2:10 core1_3:11 }
{ core1_0:null core1_1:10 core1_2:11 core1_3:12 }
\end{verbatim}

This example demonstrates directional windowing and explicit null
propagation at window boundaries, recorded in the \texttt{.meta} artifact
layer.

\subsection{Example D: Pan--Tompkins QRS Detection over MIT-BIH ECG Data}
\label{sec:ecg}

The flagship application example expresses a complete Pan--Tompkins
QRS-detection pipeline~\cite{pan1985} as a chain of algebra operators,
operating on record 205 of the MIT-BIH Arrhythmia
Database~\cite{moody2001,physionet} (two channels, 360\,Hz, declared with
the exact rational interval $1/360$\,s). All five classical stages
(bandpass filtering, differentiation, squaring, moving-window integration,
and adaptive thresholding) are expressed in RQL, with FIR convolutions
realized by the window operator \texttt{@} and the stream sum. Source:
\path{examples/ecg/rec205/rec205-qrs.rql}.

\begin{verbatim}
DECLARE MLII INTEGER, V1 INTEGER STREAM ecg, 1/360 FILE 'rec205'
DECLARE bp_coef INTEGER[25] STREAM bpf, 1 FILE 'bp_coef.txt'
DECLARE d_coef INTEGER[5] STREAM df, 1 FILE 'd_coef.txt'

# Channel extraction
SELECT ecg.MLII STREAM mlii FROM ecg VOLATILE

# 1. Bandpass filter (5-15 Hz) -- 25-tap FIR convolution
SELECT * STREAM mlii_win FROM mlii@(1,25) VOLATILE
SELECT mlii_win[_]*bpf[_] STREAM bp_acc FROM mlii_win+bpf VOLATILE
SELECT bp_acc[0]/1000 STREAM bp_out FROM bp_acc.sumc VOLATILE

# 2. Differentiation -- 5-tap FIR: [-1,-2,0,2,1]
SELECT * STREAM bp_win FROM bp_out@(1,5) VOLATILE
SELECT bp_win[_]*df[_] STREAM d_acc FROM bp_win+df VOLATILE
SELECT d_acc[0] STREAM d_out FROM d_acc.sumc VOLATILE

# 3. Squaring (division /1000 prevents int32 overflow)
SELECT d_out[0]*d_out[0]/1000 STREAM sq_out FROM d_out VOLATILE

# 4. Moving-window integration, 30 samples (~83 ms)
SELECT * STREAM mwi_win FROM sq_out@(1,30) VOLATILE
SELECT mwi_win[0] STREAM mwi FROM mwi_win.avg VOLATILE

# 5. Adaptive threshold -- moving average over 180 samples (0.5 s)
SELECT * STREAM mwi_long FROM mwi@(1,180) VOLATILE
SELECT mwi_long[0] STREAM mwi_thr FROM mwi_long.avg VOLATILE

# Output: centered MLII, envelope x5, detection signal x5
SELECT mlii[0]-900, mwi[0]*5, (mwi[0]-mwi_thr[0]*2)*5
STREAM qrs_out FROM mlii+mwi+mwi_thr VOLATILE
\end{verbatim}

\begin{figure}[t]
\centering
\includegraphics[width=0.92\linewidth]{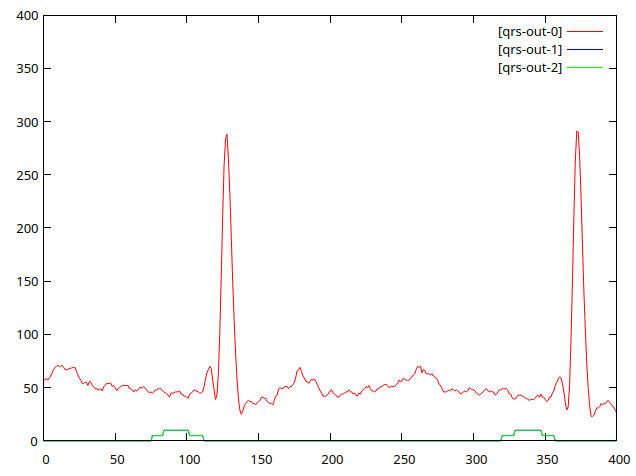}
\caption{Pan--Tompkins pipeline output for MIT-BIH record 205: centered
MLII signal, integrated envelope, and QRS detection signal, plotted live
from the \texttt{qrs\_out} stream via \texttt{xqry}.}
\label{fig:qrs}
\end{figure}

Figure~\ref{fig:qrs} shows the resulting output stream: the centered
MLII signal, the integrated envelope, and the QRS detection signal.
The pipeline demonstrates that a nontrivial, clinically standard DSP
algorithm is expressible entirely within the algebra (no user-defined
functions, no procedural escape hatches) while inheriting the
determinism guarantees of Section~\ref{sec:foundations}: replaying the
recording reproduces the detection output exactly.

\section{Discussion and Limitations}
\label{sec:limitations}

This paper focuses on formal semantics, architectural integration, and
language-level evidence.

\emph{Scope.} RetractorDB is an edge signal processing engine, not a full
time-series database: long-term retention, compression, and large-scale
ad-hoc querying are deliberately left to an upstream TSDB or DSMS that the
engine supports (Section~\ref{sec:intro}). Its data-management role is
bounded and specific: a partial, correctable record of past and
scheduled future events, held in transmissible artifacts. Its processing
role is equally specific: declarative pre-processing and filtering of
high-frequency streams before they reach the central architecture.

\emph{Irregular series.} The requirement of a constant rational $\Delta$
per stream (Definition~\ref{def:stream}) is less restrictive than it may
appear. An irregular series whose inter-arrival times are drawn from a
finite set of rationals embeds \emph{exactly} into a regular stream on
the refined grid $\Delta' = \gcd$ of those intervals, with the artifact's
null/gap metadata layer (Section~\ref{sec:system}) marking the empty
slots; for arbitrary real timestamps the embedding degrades to a
quantization with error bounded by $\Delta'/2$. Unlike zero-stuffing in
multirate DSP, the null pattern distinguishes ``no sample'' from ``sample
of value zero'', and the engine's gap entries further distinguish value
nulls from transmission losses; the run-length-encoded \texttt{.meta}
index keeps the storage cost of the embedding proportional to the number
of null-pattern runs rather than to grid density. A formal treatment of
this embedding (exactness and invertibility conditions, null-propagation
semantics across the operators of Section~\ref{sec:foundations}, and the
scheduling cost of grid refinement) is in preparation.

Beyond scope, two limitations bound the current claims.

\emph{Performance evaluation.} No throughput, latency, or cross-engine
comparison is reported here. A measurement campaign is in progress: the
engine ships a compile-time probe (\texttt{RDB\_BENCH\_PROBE}) recording
per-interval computation time on a monotonic clock, and experiments under a
real-time Linux environment (latency distributions, tail analysis, and an
exactness comparison against a floating-point reference pipeline) are
planned for a subsequent version of this work.

\emph{Timing guarantees.} The engine provides deterministic execution
\emph{semantics} (identical inputs yield identical outputs in identical
order) and a predictable, sequential execution model. We deliberately do
not claim hard real-time guarantees: those require worst-case execution-time
analysis on a real-time operating system, which is future work.

Reproducibility packaging is a further practical point: the artifact release
accompanying this paper pins the exact commit, toolchain, and data
generation procedures for all examples. Code and documentation are publicly
available~\cite{rdb_repo,rdb_docs}, enabling independent inspection of query
programs, runtime behavior, and storage outputs.

\section{Threats to Validity}

The principal construct-level threat concerns the scope of the novelty
claim. As detailed in Section~\ref{sec:related}, the underlying
number-theoretic tools are classical, and the neighboring pinwheel-scheduling
line actively publishes the same proof machinery. The claim defended here is
therefore precisely the synthesis identified in Table~\ref{tab:gap}:
covering systems as the semantic foundation of a declarative stream algebra
with exact DSP operators, realized end-to-end in a database engine,
together with the 2006 priority of that bridge~\cite{widera2006}.

On the evidence side, the examples of Section~\ref{sec:examples} constitute
semantic validation, not performance evaluation; external validity is
limited until the full language surface, adversarial workloads, and
throughput-oriented deployments are evaluated under a pinned, fully
reproducible setup. The related-work review is directed rather than
systematic (Section~\ref{sec:related}), and its conclusions are open to
revision should a use of Fraenkel's theorem in multirate DSP be identified.

\section{Conclusion}

RetractorDB demonstrates a coherent path from number-theoretic covering
systems to a deployable, deterministic edge signal processing engine that
supports, rather than replaces, time-series databases and data stream
management systems. The formal
core (interleave as a sequential covering, de-interleave as a rational
Beatty sequence satisfying Fraenkel's conditions) guarantees that exact
stream resampling is computable, invertible, and optimizable; the engine
realizes this algebra end-to-end, from declarative query semantics through
compile-time dependency planning to runtime slot scheduling and
artifact-level verification. RQL gives the engine its processing role at
the edge: declarative pre-processing and filtering of high-frequency
signals before transmission. The artifact model gives it its
data-management role: a partial, correctable record of past and scheduled
future events, ready for transmission to an upstream TSDB or DSMS.
Together they substantially offload the central architecture. The
deterministic examples, including a
complete Pan--Tompkins pipeline over MIT-BIH data, establish the semantic
evidence; a performance evaluation under real-time conditions is the
subject of ongoing work.

\section*{Acknowledgments}
This research received no institutional funding; it was conducted
independently and self-funded by the author, who declares no
conflicts of interest.


\end{document}